\documentclass[titlepage,oneside,12pt]{article}

\usepackage[tiny,rm]{titlesec}
\newpagestyle{trbstyle}{
	\sethead{Gkiotsalitis}{}{\thepage}
}
\titleformat{\section}{\bfseries}{}{0pt}{\uppercase}
\titlespacing*{\section}{0pt}{12pt}{*0}
\titleformat{\subsection}{\bfseries}{}{0pt}{}
\titlespacing*{\subsection}{0pt}{12pt}{*0}
\titleformat{\subsubsection}{\itshape}{}{0pt}{}
\titlespacing*{\subsubsection}{0pt}{12pt}{*0}

\usepackage{enumitem}
\setlist[1]{labelindent=0.5in,leftmargin=*}
\setlist[2]{labelindent=0in,leftmargin=*}

\usepackage{ccaption}
\usepackage{bigints}
\usepackage{amsmath}
\usepackage{amsfonts}
\usepackage{amssymb}
\usepackage{pgf}
\usepackage{pstricks}
\usepackage{algorithm}
\usepackage{eurosym}
\usepackage{algpseudocode}
\usepackage{mathtools}
\usepackage{import}
\usepackage{graphicx}
\usepackage{color}
\usepackage{transparent}
\usepackage{xfrac}
\usepackage{caption}
\usepackage{subcaption}
\usepackage{float}
\usepackage{palatino}
\usepackage{tikz}
\usepackage{booktabs}
\usepackage{array}
\usetikzlibrary{shapes.geometric, arrows}
\usepackage{ragged2e}
\usepackage{xfrac}  

\makeatletter
\renewcommand{\fnum@figure}{\textbf{FIGURE~\thefigure} }
\renewcommand{\fnum@table}{\textbf{TABLE~\thetable} }
\makeatother
\captiontitlefont{\bfseries \boldmath}
\captiondelim{\;}

\usepackage[sort,numbers]{natbib}
	\newcommand{\trbcite}[1]{\citeauthor{#1} ({\it \citenum{#1}})}
\setcitestyle{round}

\usepackage[pagewise]{lineno}

\usepackage{mathptmx}
\usepackage{pifont}

\usepackage[T1]{fontenc}
\usepackage{textcomp}
\usepackage{hyperref}
\hypersetup{
    colorlinks=true,
    linkcolor=blue,
    filecolor=black,      
    urlcolor=blue,
    citecolor=blue,
}

\usepackage{amssymb}
\usepackage{amsthm}
\theoremstyle{plain}
\newtheorem{theorem}{Theorem}[section]
\newtheorem{lemma}[theorem]{Lemma}

\allowdisplaybreaks
\let\OLDthebibliography\thebibliography
\renewcommand\thebibliography[1]{
  \OLDthebibliography{#1}
  \setlength{\parskip}{3pt}
  \setlength{\itemsep}{0pt plus 0.3ex}
}

\NewDocumentCommand{\floor}{s O{} m}{%
  \IfBooleanTF{#1} 
    {\left\lfloor#3\right\rfloor} 
    {#2\lfloor#3#2\rfloor} 
}

\begin{document}

	\thispagestyle{empty}

\begin{titlepage}
\begin{flushleft}

{\LARGE \bfseries Modifying the service patterns of public transport vehicles to account for the COVID-19 capacity}\\[1cm]

Dr Konstantinos Gkiotsalitis \\
Assistant Professor\\
University of Twente\\[0.2cm]
Department of Civil Engineering\\
P.O. Box 217\\
7500 AE Enschede\\
The Netherlands\\
Email: k.gkiotsalitis@utwente.nl\\ [1cm]


\today

\end{flushleft}
\end{titlepage}

\newcommand{\Spvek}[2][r]{%
  \gdef\@VORNE{1}
  \left(\hskip-\arraycolsep%
    \begin{array}{#1}\vekSp@lten{#2}\end{array}%
  \hskip-\arraycolsep\right)}

\def\vekSp@lten#1{\xvekSp@lten#1;vekL@stLine;}
\def\vekL@stLine{vekL@stLine}
\def\xvekSp@lten#1;{\def\temp{#1}%
  \ifx\temp\vekL@stLine
  \else
    \ifnum\@VORNE=1\gdef\@VORNE{0}
    \else\@arraycr\fi%
    #1%
    \expandafter\xvekSp@lten
  \fi}

\newpage

\thispagestyle{empty}
\section*{Abstract}
As public transport operators try to resume their services, they have to operate under reduced capacities due to COVID-19. Because demand can exceed capacity at different areas and across different times of the day, drivers have to refuse passenger boardings at specific stops. Towards this goal, many public transport operators have modified their service routes by avoiding to serve stops with high passenger demand at specific times of the day. Given the urgent need to develop decision support tools that can prevent the overcrowding of vehicles, this study introduces a dynamic integer nonlinear program that proposes service patterns to individual vehicles that are ready to be dispatched. In addition to the objective of satisfying the imposed vehicle capacity due to COVID-19, the proposed service pattern model caters for the waiting time of passengers. Our model is tested in a bus line connecting the university of Twente with its surrounding cities demonstrating the improvement in terms of vehicle overcrowding, and analyzing the potential negative effects related to unserved passenger demand and excessive waiting times.
 \\\\
\newline
\textbf{Keywords}: public transport; service patterns; stop-skipping; COVID-19; pandemic capacity.

\newpage
\section{Introduction}\label{sec1}
After the start of the pandemic, one country after the other implemented so-called social distancing measures affecting public transport, schools, shops, working places, and various other sectors \trbcite{anderson2020will,lewnard2020scientific}. To adjust their operations, some public transport service providers permitted the use of public transport for essential travel only (e.g., California and several other states in the US, Asia, and Europe) \trbcite{rodriguez2020going}. Given that several studies underline the increased risk of virus transmission when using public transport (see (Zhen et al., 2020)), public transport trips are considered of high transmission risk. Thus, several office workers are asked to work from home as much as possible to reduce the burden in public transport services and ensure the service availability for essential workers and vulnerable user groups.

Typical pandemic-related measures taken by public transport service providers include the limitation of the service span (e.g., not offering night services), the cancellation of certain lines, and the closure of selected stations generating new service patterns \trbcite{tirachini2020covid}. Namely, Transport for London (TfL) suspended the night tube service and closed 40 metro stations that do not interchange with other lines \trbcite{TfL2020}. Similarly, the Washington Metropolitan Area Transit Authority (WMATA) closed more than 20\% of its metro stations, reduced its service frequencies by more than half, and limited the operations of the daily metro services until 9 pm \trbcite{WMATA2020}. In Italian cities, such as Rome and Naples, services are running with reduced frequencies and early closures at 9 pm and 8 pm, respectively. Valencia in Spain has also seen a reduction in service provision of up to 35\% \trbcite{UITP2020b}. 

Implementing specific measures to ensure social distancing is the main concern of public transport operators in post-lockdown societies.
Several operators receive specific instructions from government authorities to operate under a pandemic-imposed capacity that does not allow them to use all available space inside a vehicle. This pandemic-imposed capacity aims to maintain sufficient levels of physical distancing among travelers but, at the same time, this might lead to significant numbers of unserved passengers and route/frequency changes. A recent study by \trbcite{krishnakumari2020virus} about the Washington DC metro system showed that if passengers are evenly spaced across platforms, each operating train can carry only 18\% of its nominal capacity when implementing a 1.5-meter distancing and 10\% when implementing a 2-meter distancing. \trbcite{gkiotsalitis2020optimal} showed that the average train occupancy in the Washington metro can be reduced to 11.6\%, 8.7\%, and 6.5\% when implementing 1-meter, 1.5-meter, and 2-meter social distancing policies, respectively. \trbcite{UITP2020a} also reported that to ensure necessary social distancing between 1 and 1.5-meters the transport capacity has to be reduced to 25\%-35\%, which would hardly allow accommodating travel demand. 

As part of their exit strategies, public transport authorities and operators have to devise strategies to meet the pandemic-imposed capacity limits. As of now, service providers have made adjustments to meet the pandemic capacity, but so far those adjustments are devised and implemented in an ad-hoc manner \trbcite{UITP2020b}. Typical measures include changes in service patterns by closing specific stops of the network that lead to overcrowding and adjustments of service frequencies. To rectify this, in this study we propose a dynamic service pattern model that decides about the skipped stops of every vehicle that is about to be dispatched to operate a particular trip. Our dynamic model can suggest a different service pattern for each vehicle using up-to-date passenger demand information to determine which stops should be served and which stops should be skipped for maintaining the passenger load below the pandemic-imposed capacity limit. By doing this, we exploit our vehicle resources as much as possible since we devise trip-specific service patterns instead of resorting to permanent stop closures. Besides, we maintain a level of service at all stops since stops that are not served by a vehicle will be served by a subsequent one limiting the passenger-related effects to increased waiting times. Our dynamic service pattern model can be applied every time a vehicle is about to be dispatched and can be solved in near real-time for realistic public transport lines that serve up to 60 stops.  

The remainder of this study is structured as follows: section \ref{sec2} provides a literature review on service patterns and stop-skipping models. Section \ref{sec3} introduces our integer nonlinear service pattern model and examines its computational complexity. Section \ref{sec4} provides the implementation of our model in a bus line that connects the university of Twente with its two surrounding cities using multiple passenger demand scenarios. Finally, section \ref{sec5} concludes our work and offers future research directions, including the possibility of combining our dynamic service pattern model with dynamic frequency setting models to deploy more vehicles at peak periods of the day within which the pandemic-imposed capacity cannot be maintained by merely skipping stops.

\section{Literature Review and contribution}\label{sec2}
Devising the service pattern of a vehicle at the operational level (e.g., when it is about to be dispatched) requires to determine which stops of the line should be served and which should be skipped by that vehicle (see \trbcite{li1991real,lin1995adaptive,eberlein1997real,fu2003real}). Determining the service pattern for each vehicle in isolation reduces the problem complexity and, similarly to our study, several works resort to exhaustive search methods (brute-force) to solve the dynamic problem taking advantage of the relatively small scale of the problem since typical public transport lines operate less than 40 stops \trbcite{fu2003real,sun2005real}.

The dynamic service pattern problem, which can be seen as a dynamic stop-skipping problem where all skipped stops are determined by the time the vehicle is dispatched, is typically modeled as a nonlinear integer program including assumptions of random distributions of boardings and alightings \trbcite{sun2005real}. Similar to our work, \trbcite{fu2003real} used an exhaustive search to determine the skipped stops of one trip at a time. \trbcite{fu2003real} considered the total waiting times of passengers, the in-vehicle time, and the total trip travel time as problem objectives. The potential benefit was tested with a simulation of route 7D in Waterloo, Canada. 

Different model formulations, such as \trbcite{liu2013bus}, imposed stricter stop-skipping constraints so that if a trip skips one stop, its preceding and following trip should not skip any stops. \trbcite{liu2013bus} resorted to the use of a genetic algorithm incorporating Monte Carlo simulations because of the complexity of the formulated mixed-integer nonlinear problem. \trbcite{eberlein1995phd} simplified the problem by modeling it as an integer nonlinear program with quadratic objective function and constraints enabling its analytic solution.

The dynamic service patterns problems can be solved deterministically with high accuracy because the decisions are made every time a vehicle is about to be dispatched and the up-to-date information regarding the expected passenger demand and the inter-station travel times results in low estimation errors of the realized demand and travel times. Several works, however, do not solve this problem dynamically and consider stochastic travel times and passenger demand in the problem formulation because they make decisions about future daily trips with highly uncertain travel times and passenger demand levels. \trbcite{chen2015design} used an artificial bee colony heuristic to solve the offline service pattern problem considering stochasticity since this work determined the service patterns of several vehicles ahead. \trbcite{gkiotsalitis2019robust} used also a robust optimization model for devising the service patterns of all daily trips considering the stochastic travel-times in the objective function and integrating them into a genetic algorithm that tried to find service patterns that perform reasonably at worst-case scenarios. Service patterns have also been derived for clusters of trips at the offline level \trbcite{verbas2015exploring,verbas2015stretching,gkiotsalitis2019cost}. The aforementioned works do not exploit the real-time information from telematics systems and automated passenger counts and are not in-line with the objectives of our work that needs to make well-informed decisions using up-to-date information to avoid vehicle overcrowding. Therefore, offline approaches that determine service patterns or combine service patterns with offline timetabling and offline vehicle scheduling (e.g., \trbcite{li1991real,cao2016autonomous,gao2016rescheduling,altazin2017rescheduling,cao2019autonomous}) will not be the main focus of this study. 

In past works, stop-skipping has been combined with vehicle holding (see \trbcite{eberlein1995phd,lin1995adaptive,cortes2010hybrid,
 saez2012hybrid,nesheli2015robust}
or the recent work of \trbcite{zhang2020agent}). However, such works have different objectives compared to the objectives of our model, such as improving the service regularity by implementing vehicle holding or reducing the in-vehicle travel times, which might be counterproductive when trying to maintain a pandemic-imposed capacity (e.g., when a vehicle is held more passengers will arrive at the stop and will be willing to board the vehicle leading to higher crowding levels).

From the current literature, we identify a main research gap. Whereas there is an extensive body of works on service pattern models, these works focus on improving the in-vehicle travel times, the waiting times of passengers at stops, or the service regularity. To the best of the author's knowledge, there are no works, particularly at the dynamic level, that consider the capacity limitations of vehicles when devising service patterns. Hence, they do not account for the (potential) negative effect of overcrowding that increases the risk of virus transmission. This motivates our work which proposes a novel model formulation that explicitly considers the pandemic-imposed capacity limit as the main problem objective with the use of extensive penalties for vehicles that not meet this limit when departing from a stop.

With the introduction of our model, our work addresses the following research questions:
\begin{itemize}
\item[(a)] how much can we improve the crowding levels inside vehicles to meet the pandemic-imposed capacity when changing the service patterns of vehicles?
\item[(b)] what are the side effects in terms of unserved passenger demand and increased passenger waiting times when applying such service patterns?
\end{itemize}

\section{Service Pattern Model Formulation}\label{sec3}
Our service pattern model determines the service pattern of every vehicle $n$ that is about to be dispatched. The service pattern decision is made before the vehicle is dispatched so as to inform the passengers waiting at stops about the skipped stops by that vehicle. It is important to note that when vehicle $n$ is about to be dispatched, its preceding vehicle $n-1$ is already operating its service under its pre-determined service pattern. That is, every time we determine the service pattern of a trip $n$, the service pattern of the previously dispatched vehicle $n-1$ is taken into consideration. 

The modeling part of this work relies on the following
assumptions:
\begin{itemize}
\item[(1)] The passenger arrivals at stops are random because the passengers
cannot coordinate their arrivals with the arrival times
of buses at high-frequency services \trbcite{welding1957instability,randall2007international};
\item[(2)] Passengers traveling between any origin-destination pair cannot be skipped by two consecutive vehicles, even if that means that we cannot satisfy the pandemic-imposed capacity limit.
\end{itemize}

Before proceeding to the model formulation, we introduce the following nomenclature:
\newline

\small
\noindent \begin{tabular}{p{1.2cm}p{12.2cm}}
\textbf{NOMENCLATURE}\vspace{4pt} \\
\textbf{Sets}\vspace{4pt} \\
{$S$}&{set of ordered stops of the service line, $S=\langle 1,...,s,...,|S|\rangle$;}\\\\
\textbf{Indices}\vspace{4pt} \\
$n$&vehicle that is about to be dispatched;\\
$n-1$&preceding vehicle of vehicle $n$ that is already dispatched and has a fixed service pattern;\\\\
\textbf{Parameters}\vspace{4pt} \\
$t_{n,s}$&  the expected running time of vehicle $n$ between stop $s-1$ and $s$;\\
$g_n$ & the pandemic-imposed capacity of vehicle $n$;\\
{$r_1$}&{average boarding time per passenger, a constant;}\\
{$r_2$}&{average alighting time per passenger, a constant;}\\
{$\delta$}&{average vehicle acceleration plus deceleration time for serving a stop;}\\
$\lambda_{sy}$ & the average passenger arrival rate at stop $s$ whose destination is stop $y$ (note: $\lambda_{sy}=0,~\forall 1\leq y\leq s$);\\
$M$&a very large positive number that penalizes the exceedance of the pandemic-imposed capacity\\
$\tilde{d}_{n,1}$&planned dispatching time of trip $n$;\\
$d_{n-1,s}$&expected departure time of vehicle $n-1$ from stop $s$;\\
$\tilde{x}_{n-1,s}$&$\tilde{x}_{n-1,s}=0$ if the already dispatched vehicle $n-1$ that precedes vehicle $n$ will skip stop $s\in S$, and $\tilde{x}_{n-1,s}=1$ otherwise;\\\\
\multicolumn{2}{l}{\textbf{Decision Variables}}\vspace{4pt} \\
$x_{n,s}$ & $x_{n,s}=1$ if vehicle $n$ will be allowed to serve stop $s$ and $x_{n,s}=0$ otherwise.\\\\
\multicolumn{2}{l}{\textbf{Variables}}\vspace{4pt} \\
$d_{n,s}$& the departure time of vehicle $n$ from stop $s$;\\
$a_{n,s}$&the arrival time of vehicle $n$ at stop $s$;\\
$k_{n,s}$&the dwell time of trip $n$ at stop $s$\\
$h_{n,s}$&is the time headway between the departure of trip $n-1$ and the arrival of trip $n$ at stop $s$\\
$w_{n,sy}$&the number of passengers waiting for bus $n$ and traveling from stop $s$ to $y$ (note:  $w_{n,sy}=0,~\forall y\leq s$);\\
$l_{n,sy}$& the number of passengers traveling from stop $s$ to stop $y$ skipped by vehicle $n$;\\
\end{tabular}

\noindent \begin{tabular}{p{1.2cm}p{12.2cm}}
$m_{n,s}$& the number of passengers at stop $s$ skipped by bus $n$ (note: $m_{n,s}=\sum\limits_{i=s+1}^{\mid S\mid}l_{n,si}$);\\
$u_{n,s}$& denotes the number of passengers boarding bus $n$ at stop $s$;\\
$b_{n,sy}$ & the number of passengers boarding bus $n$ at stop $s$ whose destination is stop $y$ (note: $b_{n,sy}=0,~\forall y\leq s$);\\
$\nu_{n,s}$ &the number of passengers alighting bus $n$ at stop $s$, where $n\in N, s\in S$ (note: $\nu_{n,1}=0,~\forall n\in N$);\\
$\mu_{s} $ & the average passenger arrival rate at stop $s$ (note: $\mu_s=\sum\limits_{i=s+1}^{\mid S\mid}\lambda_{si}$).\\ 
$\gamma_{n,s}$ &the passenger load of vehicle $n$ when traveling from stop $s$ to stop $s+1$.\\ 
\end{tabular}
\newline\newline

\normalsize

The service pattern model that decides about the skipped stops by vehicle $n$ is presented in Eqs.\eqref{eqobj}-\eqref{eq20}. Our model is based on the vehicle movement ideas expressed in \trbcite{fu2003real,gkiotsalitis2020stop}, which are expanded to cater for the pandemic capacity limit and to be able to make dynamic service pattern decisions for each individual vehicle.

\begin{align}
\min f(\textbf{x}):= &~M\sum_{s=1}^{|S|-1}\max\{0,(\gamma_{n,s}-g_n)\}+\nonumber\\&~~~~~~~~\sum_{s=1}^{|S|-1}\Big[ (u_{n,s}-m_{n-1,s})\frac{h_{n,s}}{2}+m_{n,s}\Big(\frac{h_{n,s}}{2}+k_{n,s}+h_{n+1,s}\Big)\Big] \label{eqobj}\\
\text{subject to: }~l_{n,sy}=& \begin{cases}0,\text{ if }y\leq s\\
w_{n,sy}-w_{n,sy}x_{n,s}x_{n,y},\text{ if }y> s
\end{cases} \label{eq1}\\
m_{n,s}=& \sum_{y=s+1}^{|S|}l_{n,sy},~\forall s\in S\setminus\{|S|\}\label{eq2}\\
w_{n,sy}=& 
l_{n-1,sy}+\lambda_{sy}h_{n,s},~\forall s\in S\setminus\{|S|\} \label{eq3}\\ 
u_{n,s}=& x_{n,s}\sum_{y=s+1}^{|S|}w_{n,sy}x_{n,y},~\forall s\in S\setminus\{|S|\})\label{eq4}\\
b_{n,sy}=& \begin{cases}
x_{n,s}w_{n,sy}x_{n,y},\text{ if }y>s\\
0,\text{ if }y\leq s
\end{cases}\label{eq6}\\
\nu_{n,s}=& x_{n,s} \sum\limits_{y=1}^{s-1}w_{n,ys}x_{n,y},~\forall s\in S\setminus\{1\} \label{eq7}\\
k_{n,s}=& \max\big(r_{1} u_{n,s},r_{2} \nu_{n,s}\big) \label{eq9}\\
a_{n,s}=&\begin{cases} d_{n,s-1}+t_{n,s}+\frac{\delta}{2}(x_{n,s-1}+x_{n,s}),~\forall s\in S\setminus\{1,2\}\\
\tilde{d}_{n,s-1}+t_{n,s}+\frac{\delta}{2}(x_{n,s-1}+x_{n,s}),~s=2\end{cases} \label{eq10}\\
d_{n,s}=& a_{n,s}+k_{n,s},~\forall s\in S\setminus\{1\} \label{eq12}\\
h_{n,s}=& \begin{cases}a_{n,s}-d_{n-1,s},~\forall s\in S\setminus\{1\}\\
\tilde{d}_{n,s}-\tilde{d}_{n-1,s},~s=1\end{cases}\label{eq13}\\
x_{n,1}=&1\label{eq18}\\
x_{n,|S|}=&1\label{eq18b}\\
(\tilde{x}_{n-1,s}\tilde{x}_{n-1,y})+(x_{n,s}x_{n,y})\geq& 1,~\forall s\in S,~\forall y\geq s\label{eq19}\\
 x_{n,s}\in \{0,1\}~&,~\forall s\in S\label{eq20}
\end{align}

The objective function \eqref{eqobj} consists of two components. The first component, $$M\sum_{s=1}^{|S|-1}\max\{0,(\gamma_{n,s}-g_n)\}$$ penalizes in-vehicle passenger loads $\gamma_{n,s}$ that are higher that the pandemic-imposed capacity $g_n$. Namely, if at stop $s\in S$ the passenger load is lower than $g_n$, then $\gamma_{n,s}-g_n<0$ and $M\max\{0,(\gamma_{n,s}-g_n)\}=0$ leading to no additional penalty to the objective function. If, however, at stop $s$ the load is greater than capacity $g_n$, then $\gamma_{n,s}-g_n>0$ and $M\max\{0,(\gamma_{n,s}-g_n)\}=M(\gamma_{n,s}-g_n)$ penalizing the objective function with the excessive penalty $M(\gamma_{n,s}-g_n)$. This excessive penalization directs the objective function towards finding service patterns that do not lead to loads that exceed the pandemic-imposed capacity. 

The second component of the objective function, $$\sum_{s=1}^{|S|-1}\Big[ (u_{n,s}-m_{n-1,s})\frac{h_{n,s}}{2}+m_{n,s}\Big(\frac{h_{n,s}}{2}+k_{n,s}+h_{n+1,s}\Big)\Big]$$ computes the total waiting time of passengers who arrive after the departure of vehicle $n-1$ from stop $s$, ($(u_{n,s}-m_{n-1,s})\frac{h_{n,s}}{2}$), plus the total waiting time of those passengers who have been stranded by vehicle $n$ ($m_{n,s}$) and have to wait for an average amount of time equal to $m_{n,s}\Big(\frac{h_{n,s}}{2}+k_{n,s}+h_{n+1,s}\Big)$ for the following vehicle $n+1$. Note that this formulation implies that a stranded passenger by vehicle $n$ will be always served by the next vehicle, $n+1$, which will be enforced by constraint \eqref{eq19} when we have to decide about the service pattern of vehicle $n+1$.

Constraint \eqref{eq1} returns the number of passengers destined to stop $y$ who are stranded by vehicle $n$ at stop $s, (l_{n, sy})$. $l_{n, sy}$ is equal to 0 if vehicle $n$ serves stops $s$ and $y$. Otherwise, $l_{n, sy}$ is equal to the number of passengers waiting for vehicle $n$
at stop $s$ and have stop $y>s$ as their destination. Constraint \eqref{eq2} returns the number of passengers at stop $s$ skipped by vehicle $n$. Constraint \eqref{eq3} returns the number of passengers waiting for vehicle $n$ at stop $s$ whose destination is stop $y>s$.

Constraint \eqref{eq4} returns the expected number of passengers who will board vehicle $n$ at stop $s$ (assuming that stop $s$ belongs to its service pattern) and depends on the number of passengers traveling between stops $s$ and $y>s$ and whether the vehicle will stop at stop $y$. Constraint \eqref{eq6} returns the total amount of passengers boarding vehicle $n$ at stop $s$ whose destination is stop $y>s$. Constraint \eqref{eq7} returns the expected number of alighting passengers at stop $s$ depending on whether vehicle $n$ will serve that stop. Constraint \eqref{eq9} returns the dwell time at stop $s$ and depends on the number of passengers that will board and alight at the stop, denoted by $u_{n,s}$ and $\nu_{n,s}$, respectively. Note that the dwell time constraint implies that passengers use different door channels for boardings and alightings because we do not want high concentration of passengers at the same door channel.

Constraint \eqref{eq10} returns the arrival time of vehicle $n$ at every stop depending on the inter-station travel times and the service pattern choice. Constraint \eqref{eq12} returns the departure time from stop $s$ by adding the dwell time to the arrival time. Constraint \eqref{eq13} returns the headway between vehicles $n$ and $n-1$ at each stop. This headway is the time difference between the arrival time of vehicle $n$ at stop $s$ and the departure time of its preceding trip $n-1$ from the same stop.
 
Constraints \eqref{eq18}-\eqref{eq18b} ensure that our service pattern will include the first and last stops of the line. Constraint \eqref{eq19} ensures that passengers traveling between any origin-destination pair cannot be skipped by two consecutive vehicles (see assumption (2)). Finally, constraint \eqref{eq20} ensures that our decision variable is binary (that is, we can either skip or serve a stop $s$).

Our service pattern model expressed in Eqs.\eqref{eqobj}-\eqref{eq20} considers the pandemic-imposed capacity limit in its objective function and it is an integer nonlinear programming problem (INLP) that needs to be solved every time a new vehicle is about to be dispatched. Due to its combinatorial nature, the problem can be solved to \textit{global optimality} with exhaustive search of the solution space. Exploring the entire solution space with brute force results in exponential computational complexity. This is formally proved in lemma \ref{theorem1}.

\begin{lemma}\label{theorem1}
The dynamic service pattern problem subject to a pandemic-imposed capacity has an exponential computational complexity that requires to explore $2^{|S|-2}$ potential solutions.
\end{lemma}
\begin{proof}
The decision variables that determine the service pattern, $x_{n,s}$, can receive binary values. Formally, to find the globally optimal service pattern of a vehicle $n$, we need to explore a set of $2^{|S|}$ potential solutions because at each stop $s\in\{1,2,...,|S|\}$ we have two options: serve or skip ($x_{n,s}\in\{0,1\}$). Ergo, the potential solutions that need to be evaluated are $2^{|S|}$. Given that the first and the last stop of the line should be always part of the service pattern, this solution space is reduced to $2^{|S|-2}$.
\end{proof}

At this point we should note that, as also demonstrated in \trbcite{fu2003real}, solving this problem to global optimality with brute-force is possible in public transport lines with realistic sizes. If, however, the number of stops is unrealistically high then we cannot evaluate every potential service pattern with the use of brute-force and we need to resort to sub-optimal solution approximations with the use of heuristics. Fig.\ref{fig2} provides an indication of the increase of the solution space with the number of stops indicating also the possible computations that can be executed by the world's fastest supercomputer that can execute up to 33,860 trillion calculations per second. We note that we select the computation limit of one minute because the decision about the service patter of each vehicle should be made within a limited time during which the vehicle awaits to be dispatched. Under this condition, a solution can be computed for lines with up to 60 stops (see Fig.\ref{fig2}).

\begin{figure}[H]
\centering
\includegraphics[scale=1.0]{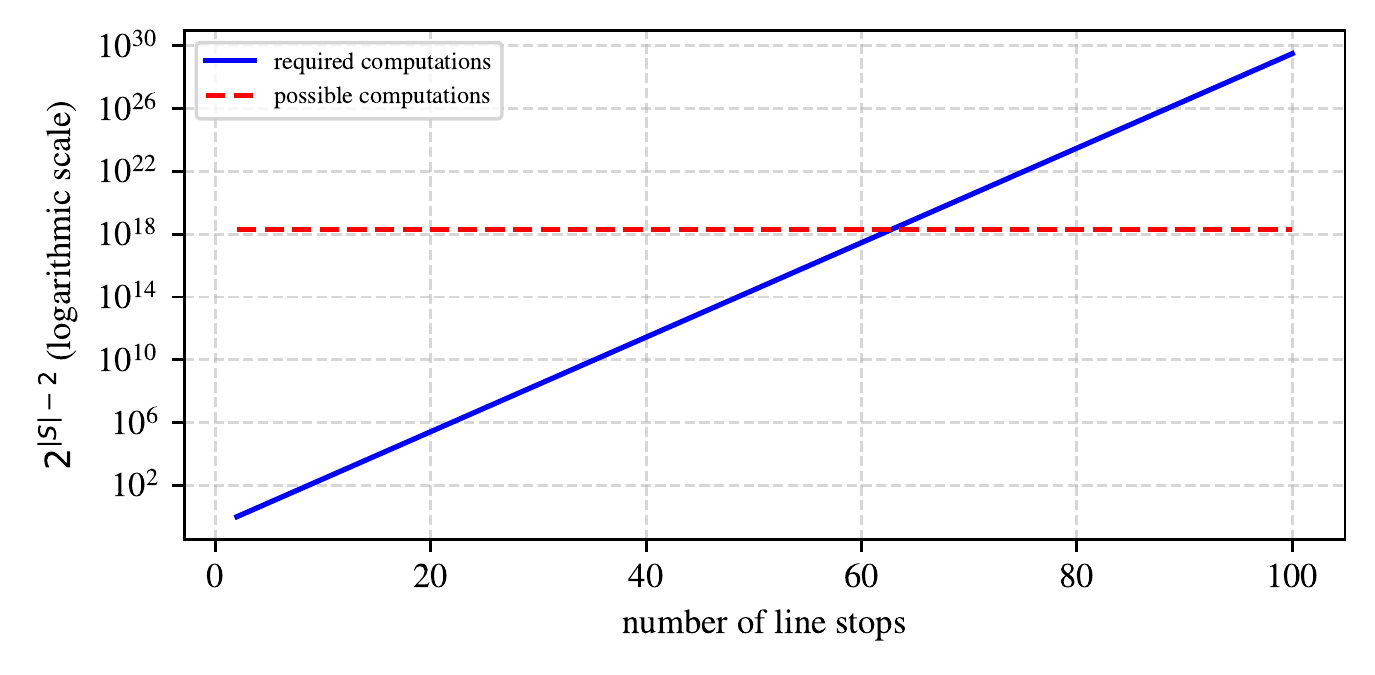}
\caption{Required solution evaluations when the size of the line, $|S|$, varies. \label{fig2}}
\end{figure}

\section{Numerical Experiments}\label{sec4}
\subsection{Case study description}
Our case study is bus line 9 in the Twente region operated by Keolis Nederland. The bus line connects two cities: Hengelo with 80 thousand inhabitants and Enschede with 160 thousand inhabitants. This line is selected because it also serves the university of Twente (UT) which is approximately in the middle of the two cities. The line consists of 13 stops per direction and the stops that accommodate the university (Enschede Kennispark / UT and Enschede Westerbegraafplaats / UT) are the 7th and 8th when considering the direction from Hengelo to Enschede. The topology of the line is presented in Fig.\ref{fig1}.

\begin{figure}[H]
\centering
\includegraphics[scale=1.2]{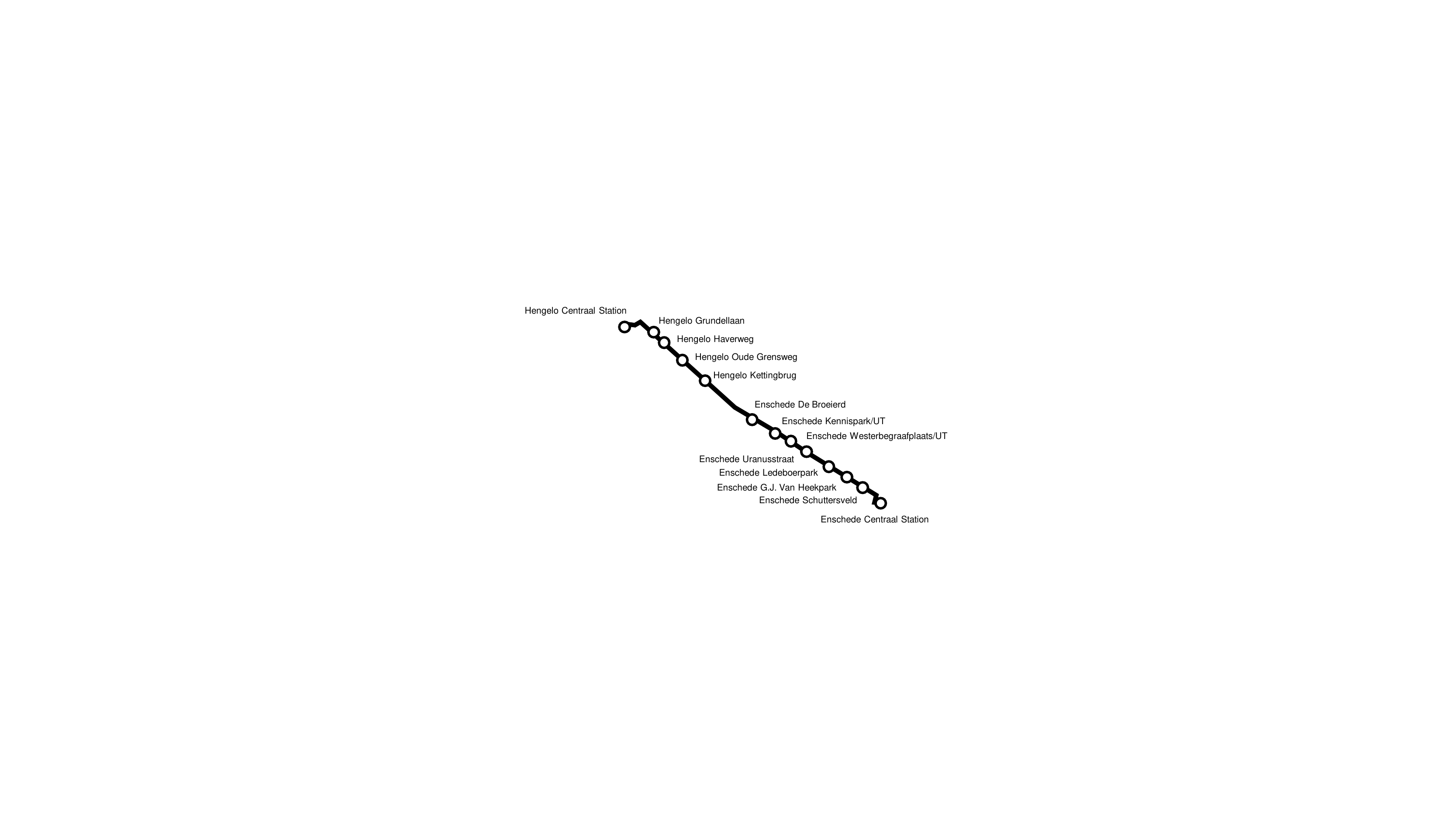}
\caption{Topology of bus line 9 in Twente, Netherlands\label{fig1}}
\end{figure}

The average trip travel time per direction is 16 minutes. In addition, the operational hours of the line service are presented in Table \ref{tab1}.

\begin{table}[H]
\centering
\caption{Operational hours of bus line 9 (line direction Hengelo centraal - Enschede centraal)\label{tab1}}
\begin{tabular}{lr}
\toprule
Day & Operational hours\\\midrule
Monday & 06:29 - 23:29\\
Tuesday & 06:29 - 23:29\\
Wednesday & 06:29 - 23:29\\
Thursday & 06:29 - 23:29\\
Friday & 06:29 - 23:29\\
Saturday & 07:29 - 23:29\\
Sunday & 10:29 - 23:29\\
\bottomrule
\end{tabular}
\end{table}

The average boarding time per passenger is $r_1=2$ seconds and the average alighting time $r_2=1$ second. In addition, when serving a stop a bus needs an average time of $\delta=20$ seconds for acceleration/deceleration.

\subsection{Evaluation}

Because the travel demand can vary across time, we generate a large number of passenger demand scenarios in our numerical experiments. For this, we focus at the time period 8:00-9:00 and we use the origin-destination demand presented in Table \ref{tab2} as follows: we generate 1000 origin-destination demand scenarios by sampling demand values from a \textit{restricted} normal distribution that restricts the sampling of negative demand values. This normal distribution uses as mean the presented hourly demand values in Table \ref{tab2}. As in \trbcite{welding1957instability,randall2007international}, it is hypothesized that the passenger arrivals at stops are random, following a uniform distribution because of the high service frequency (see assumption 1). During this time period, the planned headway among successive vehicles is 5 minutes. The pandemic-imposed capacity is $g_n=25$ passengers and, given that the vehicles have 33 seats, we assume a nominal vehicle capacity of 43 passengers including 10 standees.

\begin{table}[H]
  \centering
  \caption{Hypothesized hourly origin-destination matrix from 8:00-9:00 from which we generate one thousand origin-destination demand scenarios\label{tab2}}
    \begin{tabular}{lrrrrrrrrrrrrr}
    \toprule
    &\multicolumn{13}{c}{destination stops}\\\cmidrule{2-14}
      origin stops   & 1     & 2     & 3     & 4     & 5     & 6     & 7     & 8     & 9     & 10    & 11    & 12    & 13 \\\midrule
  1&  \ding{55}     & 4     & 8     & 8     & 12    & 12    & 8     & 8     & 12    & 8     & 4     & 16    & 22 \\
 2&   \ding{55}     & \ding{55}     & 2     & 4     & 4     & 8     & 8     & 4     & 8     & 12    & 20    & 12    & 26 \\
 3&   \ding{55}     & \ding{55}     & \ding{55}     & 2     & 2     & 2     & 16    & 12    & 8     & 8     & 16    & 20    & 16 \\
  4&  \ding{55}     & \ding{55}     & \ding{55}     & \ding{55}     & 4     & 4     & 8     & 12    & 16    & 20    & 4     & 12    & 28 \\
  5&  \ding{55}     & \ding{55}     & \ding{55}     & \ding{55}     & \ding{55}     & 2     & 4     & 4     & 10    & 10    & 4     & 14    & 14 \\
 6&   \ding{55}     & \ding{55}     & \ding{55}     & \ding{55}     & \ding{55}     & \ding{55}     & 4     & 2     & 4     & 12    & 6     & 10    & 16 \\
  7&  \ding{55}     & \ding{55}     & \ding{55}     & \ding{55}     & \ding{55}     & \ding{55}     & \ding{55}     & 2     & 2     & 2     & 6     & 12    & 24 \\
  8&  \ding{55}     & \ding{55}     & \ding{55}     & \ding{55}     & \ding{55}     & \ding{55}     & \ding{55}     & \ding{55}     & 4     & 4     & 2     & 4     & 18 \\
  9&  \ding{55}     & \ding{55}     & \ding{55}     & \ding{55}     & \ding{55}     & \ding{55}     & \ding{55}     & \ding{55}     & \ding{55}     & 2     & 4     & 8     & 22 \\
  10&  \ding{55}     & \ding{55}     & \ding{55}     & \ding{55}     & \ding{55}     & \ding{55}     & \ding{55}     & \ding{55}     & \ding{55}     & \ding{55}     & 2     & 6     & 24 \\
  11&  \ding{55}     & \ding{55}     & \ding{55}     & \ding{55}     & \ding{55}     & \ding{55}     & \ding{55}     & \ding{55}     & \ding{55}     & \ding{55}     & \ding{55}     & 4     & 6 \\
  12&  \ding{55}     & \ding{55}     & \ding{55}     & \ding{55}     & \ding{55}     & \ding{55}     & \ding{55}     & \ding{55}     & \ding{55}     & \ding{55}     & \ding{55}     & \ding{55}     & 2 \\
   13& \ding{55}     & \ding{55}     & \ding{55}     & \ding{55}     & \ding{55}     & \ding{55}     & \ding{55}     & \ding{55}     & \ding{55}     & \ding{55}     & \ding{55}     & \ding{55}     & \ding{55} \\
    \bottomrule
    \end{tabular}%
\end{table}%

In addition, we select a particular trip $n$ which is about to be dispatched and for which we should determine a service pattern. Note that its preceding trip $n-1$ serves all stops and, following the line schedule, the dispatching headway between trips $n$ and $n-1$ is 5 minutes. Note also that we run 1000 experiments with different origin-destination demand scenarios to obtain less biased results from several potential demand realizations. In our experiments, we compare the performance of the following models:

\begin{itemize}
\item[(a)] the \textit{as-is} case where trip $n$ serves all stops without applying a service pattern;
\item[(b)] the \textit{SPM - nominal capacity} service pattern model that does not consider the pandemic-imposed vehicle capacity by solving the model in Eqs.\eqref{eqobj}-\eqref{eq20} without considering the first term of the objective function, $M\sum_{s=1}^{|S|-1}\max\{0,(\gamma_{n,s}-g_n)\}$, and enforcing the nominal vehicle capacity of 43 passengers;
\item[(c)] the proposed \textit{SPS - pandemic capacity} service pattern model that considers the pandemic imposed vehicle capacity expressed in Eqs.\eqref{eqobj}-\eqref{eq20}.
\end{itemize}

To evaluate the performance of the aforementioned methods, we use the following key performance indicators:

\begin{itemize}
\item[(i)] $\mathcal{O}_1$: the total passenger load of vehicle $n$ that exceeds the pandemic-imposed capacity at all stops;
\item[(ii)] $\mathcal{O}_2$: the number of unserved passengers by vehicle $n$ that have to wait for the next vehicle;
\item[(iii)] $\mathcal{O}_3$: the total waiting time of passengers that will have to board the next vehicle, $n+1$.
\end{itemize}

The aforementioned key performance indicators will assess the benefits of our proposed \textit{SPS - pandemic capacity} model, while evaluating also its potential externalities, such as the increase in unserved passengers and the increase in passenger waiting times at stops for passengers that are refused to board vehicle $n$.

Our \textit{SPS - pandemic capacity} is programmed in Python 3.7 using a general-purpose computer with Intel Core i7-7700HQ CPU @ 2.80GHz and 16 GB RAM. The \textit{SPS - nominal capacity} is also programmed in Python 3.7 and its main difference is that it does not consider the pandemic capacity term $$M\sum\limits_{s=1}^{|S|-1}\max\{0,(\gamma_{n,s}-g_n)\}$$ in its objective function. For each one of the 1000 passenger demand scenarios, the optimal service pattern solutions of the \textit{SPS - pandemic capacity} and the \textit{SPS - nominal capacity} models are computed with brute-force by evaluating all possible service pattern alternatives. The most typical service pattern solutions that appeared in most of the scenarios are presented in Fig.\ref{fig3}.

\begin{figure}[H]
\centering
\includegraphics[scale=0.9]{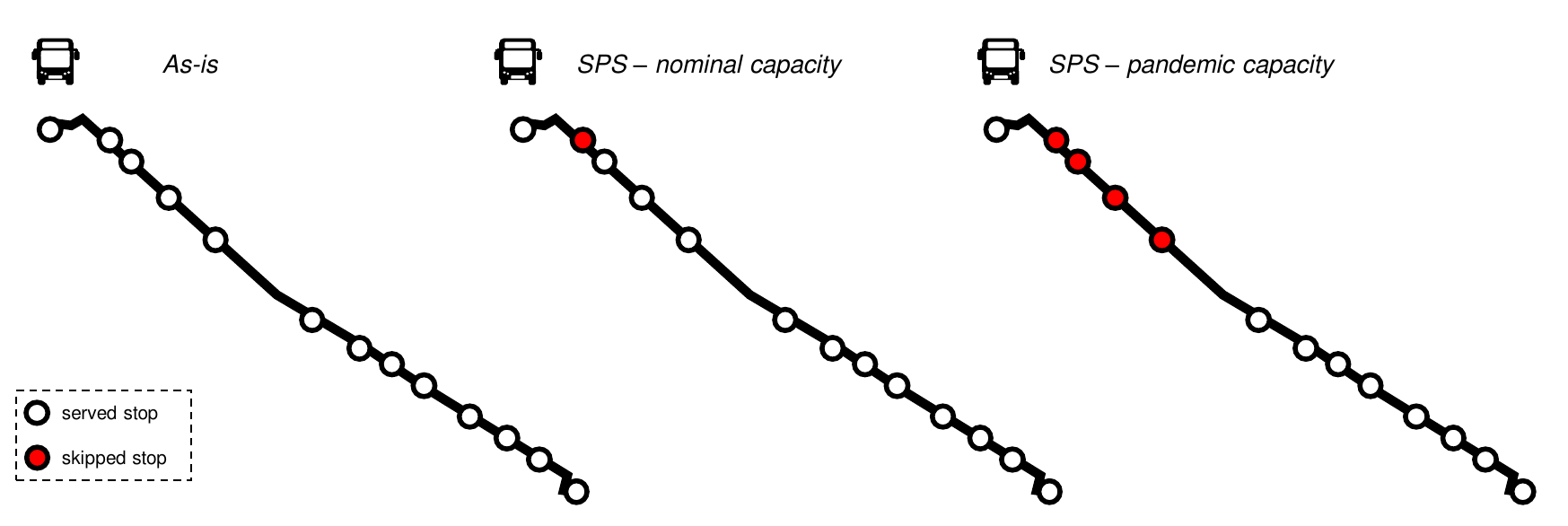}
\caption{Typically suggested service patterns based on each method. Note that \textit{SPS - pandemic capacity} results in an increased number of skipped stops.\label{fig3}}
\end{figure}

First, the results of our implementation with respect to the total in-vehicle passenger load that exceeds the pandemic-imposed capacity, $\mathcal{O}_1$, are presented in Fig.\ref{fig_boxplot_load}. Note that this key performance metric sums the number of passengers that exceed the pandemic-imposed capacity when the vehicle departs from every stop. Therefore, the same passenger can be counted more than once if he/she remains at vehicle $n$ at several stops while the pandemic-imposed capacity is exceeded. In Fig.\ref{fig_boxplot_load} we use the Tukey boxplot convention (see \trbcite{mcgill1978variations}) to present the results from the 1000 scenarios. In the Tukey boxplot, the upper and lower boundaries of the boxes indicate the upper and lower quartiles (i.e. 75th and 25th percentiles denoted as Q$_3$ and Q$_1$, respectively). The lines vertical to the boxes (whiskers) show the maximum and minimum values that are not outliers. The whiskers are determined by plotting the lowest datum still within 1.5 the interquartile range (IQR) Q$_3$-Q$_1$ of the lower quartile, and the highest datum still within 1.5 IQR of the upper quartile. 

\begin{figure}[H]
\centering
\includegraphics[scale=1]{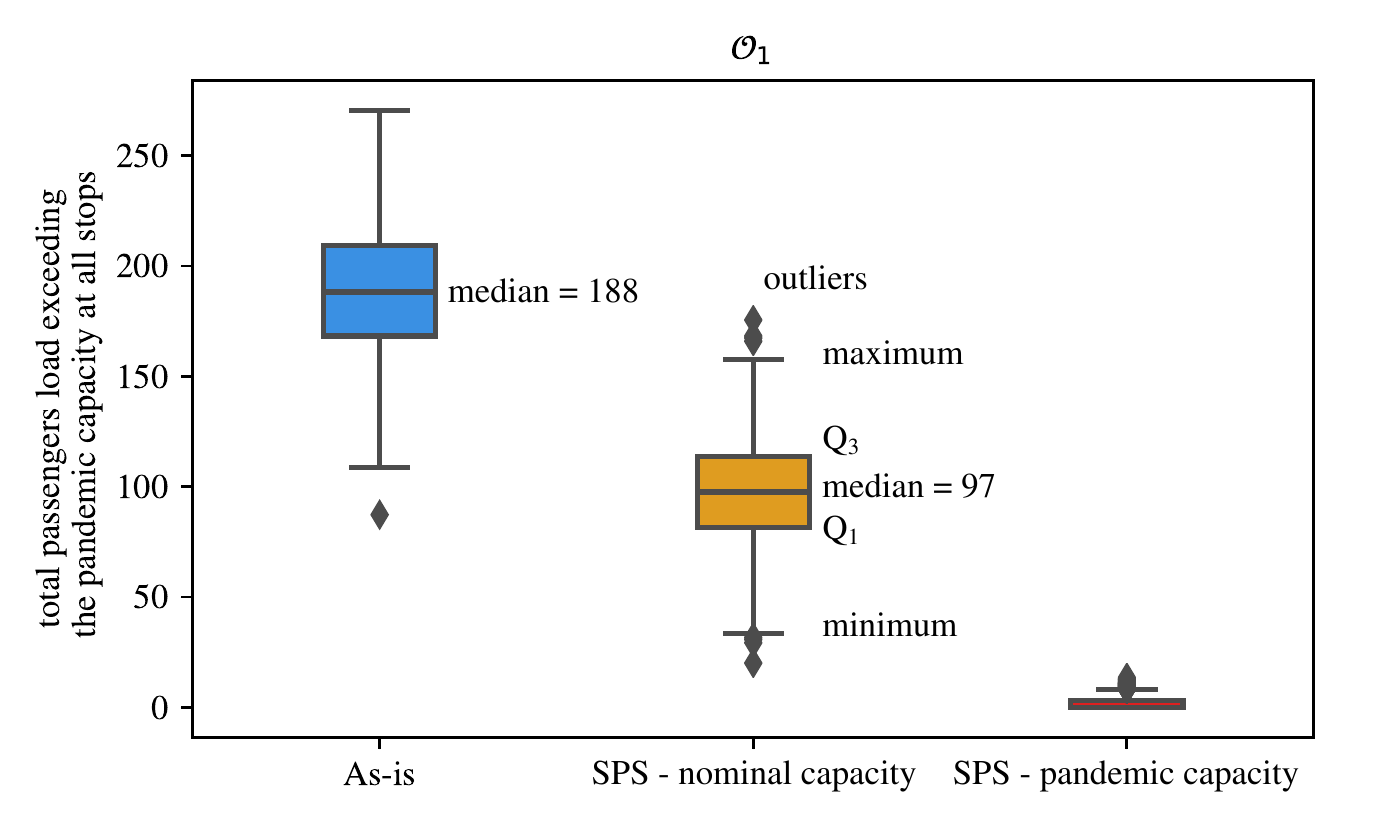}
\caption{Tukey boxplot of the performances of the \textit{as-is}, \textit{SPS - nominal capacity}, and \textit{SPS - pandemic capacity} designs in 1000 scenarios\label{fig_boxplot_load}}
\end{figure}

From Fig.\ref{fig_boxplot_load} one can note that the \textit{SPS - pandemic capacity} service pattern results in passenger loads that do not exceed the pandemic-imposed capacity in the vast majority of cases. This was partly expected because the objective function of this model penalizes heavily service patterns that might result in crowding levels above the pandemic-imposed capacity. The worst performance is observed at the \textit{as-is} design which does not skip any stops resulting in a median passenger load above the pandemic-imposed capacity of 188 passengers when considering the sum over all stops. One interesting observation is that the \textit{SPS - nominal capacity} service pattern, which only skips one stop (stop 2), improves this key performance indicator by almost 50\% compared to the as-is service pattern.

Besides this aggregate analysis, it is also important to investigate which stops are problematic leading to excessive passenger loads inside vehicle $n$. For this, we also report in Fig.\ref{fig_load_stations} the average values of the passenger loads from the 1000 scenarios after the departure of vehicle $n$ from each stop $s\in\langle 1,2,...13\rangle$.

\begin{figure}[H]
\centering
\includegraphics[scale=1]{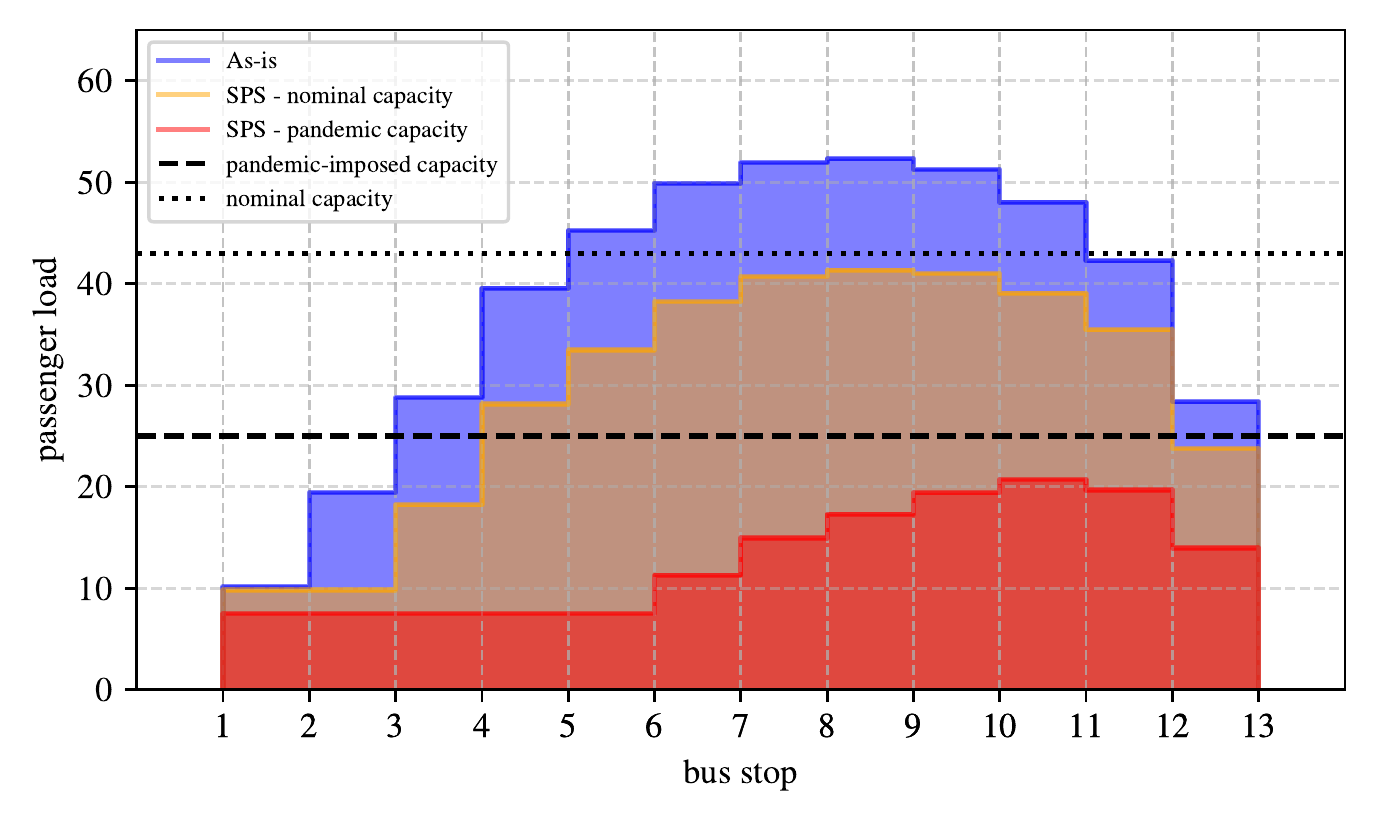}
\caption{Average passenger load when departing from each stop when implementing the service pattern with nominal capacity, with pandemic capacity, and the as-is design\label{fig_load_stations}}
\end{figure}

The results of Fig.\ref{fig_load_stations} help us to interpret the \textit{SPS - pandemic capacity} solution presented in Fig.\ref{fig3} that skips stops 2, 3, 4 and 5. From Fig.\ref{fig_load_stations} one can note that the vehicle load increases significantly due to the passenger boardings at stops 2, 3, 4 and 5. After stop 6, the passenger load increases are marginal and the passenger load starts to decrease for all designs after stop 10. Both the \textit{as-is} service pattern and the service pattern of the \textit{SPS - nominal capacity} exceed the pandemic-imposed capacity after departing from stop 4. In contrast, skipping stops 2, 3, 4 and 5 that exhibit the highest passenger inflows allow the \textit{SPS - pandemic capacity} to maintain a passenger load below the pandemic-imposed capacity at all stops. 

It is important to note that the as-is case results in passenger loads of more than 50 passengers for stops 6-9 increasing considerably the risk of virus transmission. In addition, this design exceeds the pandemic-imposed capacity limit when traveling from stop 3 until the end of the line. The \textit{SPS - nominal capacity} service pattern has an improved performance and exceeds the pandemic-imposed capacity limit when traveling from stop 4 until stop 12. It also results at more than 40 onboard passengers at the segment comprising stops 7-9. However, as expected, it remains below the nominal capacity limit of 43 passengers.

After the performance evaluation of the service patterns for the first key performance indicator $\mathcal{O}_1$ it is clear that our proposed \textit{SPS - pandemic capacity} performs significantly better than the alternative designs. This, however, comes at a cost because the \textit{SPS - pandemic capacity} service pattern skips four stops resulting in increased passenger waiting times and unserved demand. The number of unserved passengers might be an issue because we do not only refuse service to passengers who want to board at the skipped stops 2, 3, 4 and 5, but also to passengers that want to alight at these stops. To compare the numbers of unserved passengers by vehicle $n$, we implement the three service patterns in the same 1000 scenarios with different origin-destination demand and we report the aggregated results in the boxplots of Fig.\ref{fig_unserved_demand}.

\begin{figure}[H]
\centering
\includegraphics[scale=1]{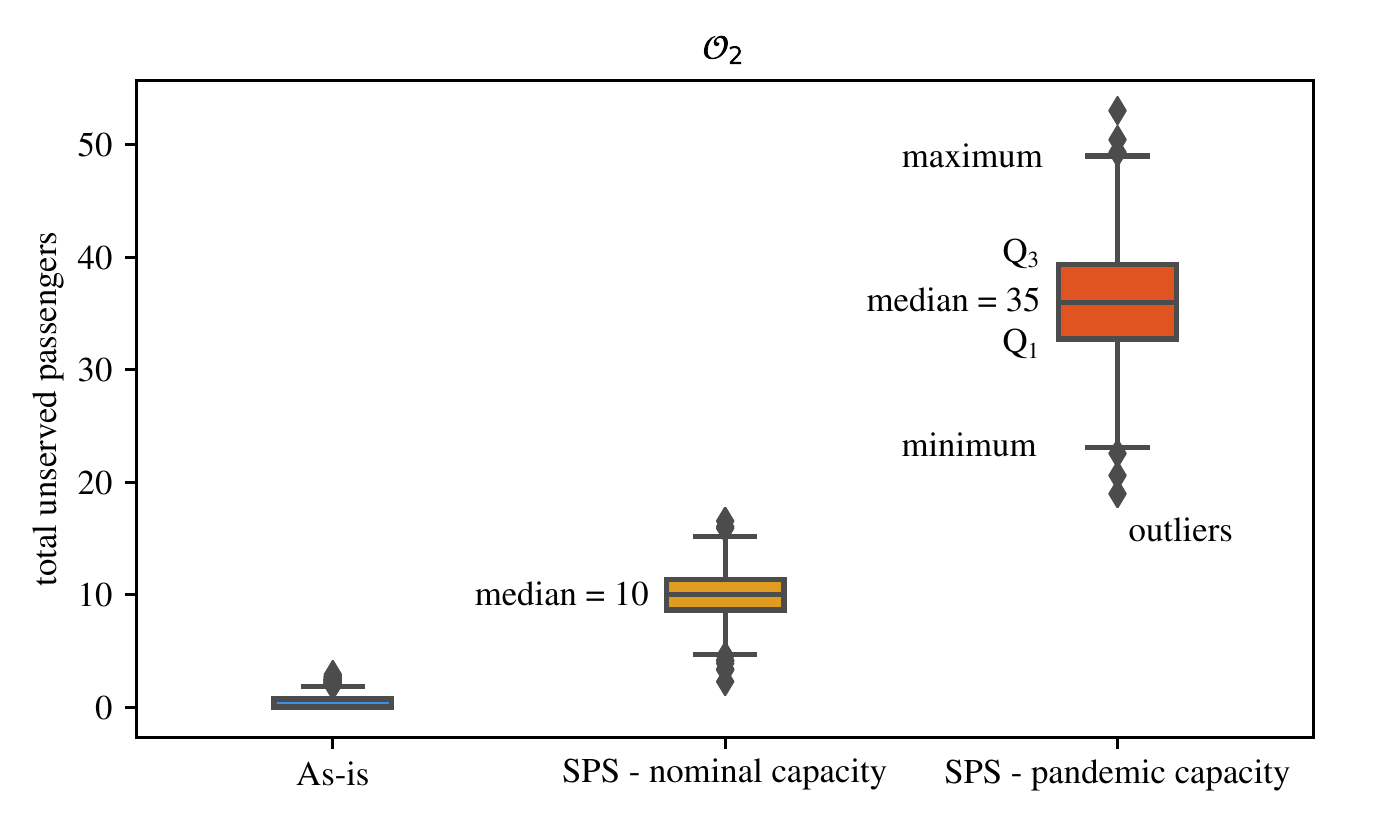}
\caption{Tukey boxplot of the unserved demand by vehicle $n$ for each service pattern implemented in 1000 scenarios with random passenger demand\label{fig_unserved_demand}}
\end{figure}

Interestingly, one can note that the \textit{SPS - pandemic capacity} service pattern results in a median of 35 unserved passengers by vehicle $n$ that can increase to up to 50 passengers in worst-case demand scenarios. This is a very useful result for public transport operators because it will help them to quantify the cost of unserved demand when trying to maintain the pandemic-imposed vehicle capacity. It is also of interest to understand which are the most problematic stops where passengers are refused boarding since this might lead to crowding at the waiting platforms. For this, we also report the average number of unserved passengers at each stop in Fig.\ref{fig_stranded_passengers}.

\begin{figure}[H]
\centering
\includegraphics[scale=1]{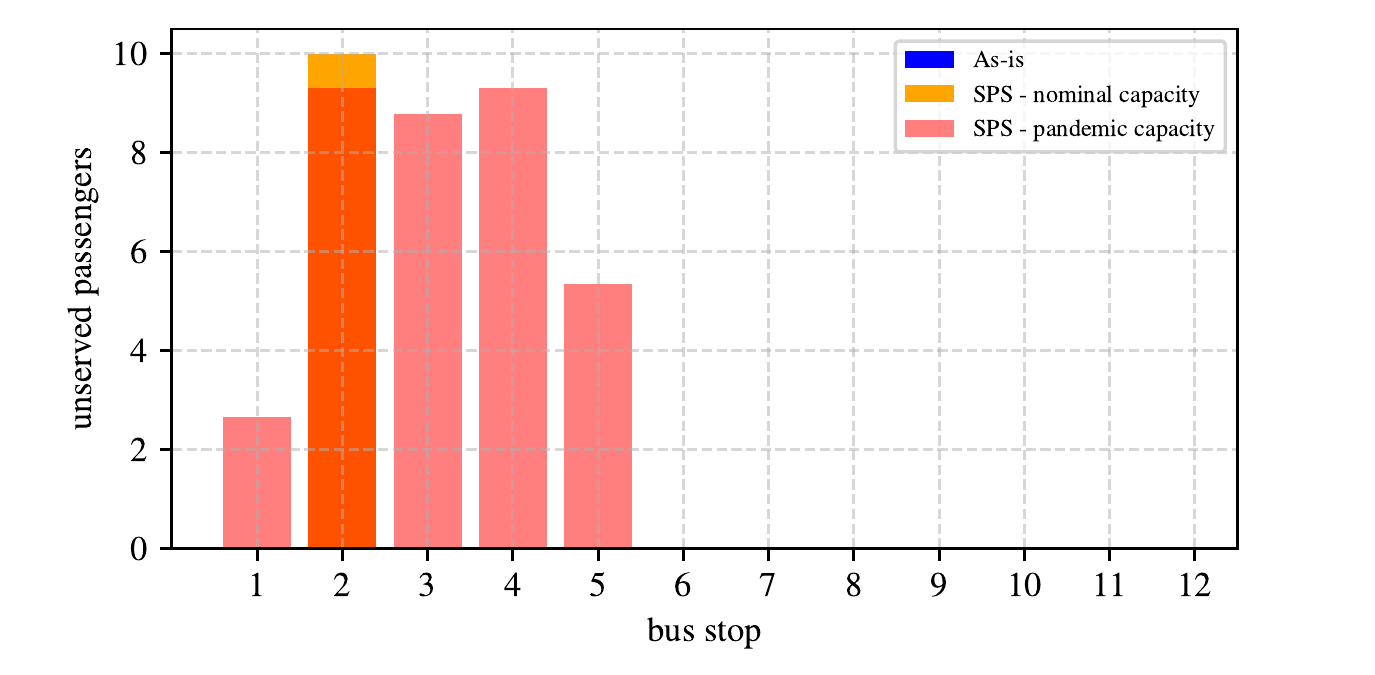}
\caption{Average number of unserved passengers by vehicle $n$ at each stop when implementing the service pattern with nominal capacity, with pandemic capacity, and the as-is design\label{fig_stranded_passengers}}
\end{figure}

Note that for the $\textit{SPS - pandemic capacity}$ case we have more than two unserved passengers at stop 1 even if stop 1 is not skipped by vehicle $n$, as can be seen in Fig.\ref{fig3}. The reason behind this is that passengers who want to alight at stops 2, 3, or 4 do not board vehicle $n$ because it skips those stops. Additionally, the number of unserved passengers in Fig.\ref{fig_stranded_passengers} appears as a real number although it should be a positive integer. The reason is that Fig.\ref{fig_stranded_passengers} reports the average value of unserved passengers by vehicle $n$ from the 1000 scenarios.

From Fig.\ref{fig_stranded_passengers} one can note that there are no unserved passengers after stop 6. This is expected since all designs do not skip any stops in the segment comprised of stops 6-12. The last stop of the line (13) is not plotted because it cannot be skipped. When implementing the \textit{SPS - nominal capacity} service pattern, all unserved passengers are at stop 2 - which is the only skipped stop by this service pattern. When implementing the \textit{SPS - pandemic capacity} though, we have unserved passengers at more stops (stops 1-5). This demonstrates that this service pattern does not only lead to more unserved passengers to satisfy the pandemic-imposed capacity limit, but it also leads to unserved passenger demand at more stops. 

The increased passenger waiting times expressed in the key performance indicator $\mathcal{O}_3$ are mostly analogous to the results of the unserved demand because the more unserved passengers, the higher their total waiting time. The results with respect to the key performance indicator $\mathcal{O}_3$ are presented in Fig.\ref{fig_waiting_times} and, as expected, are mostly in-line with the results of Fig.\ref{fig_unserved_demand}. The median of the total waiting time of unserved passengers is 180 minutes when implementing the service pattern \textit{SPS - pandemic capacity}. This corresponds to a waiting time of (approximately) 5.1 minutes for each unserved passenger by vehicle $n$. 

\begin{figure}[H]
\centering
\includegraphics[scale=1]{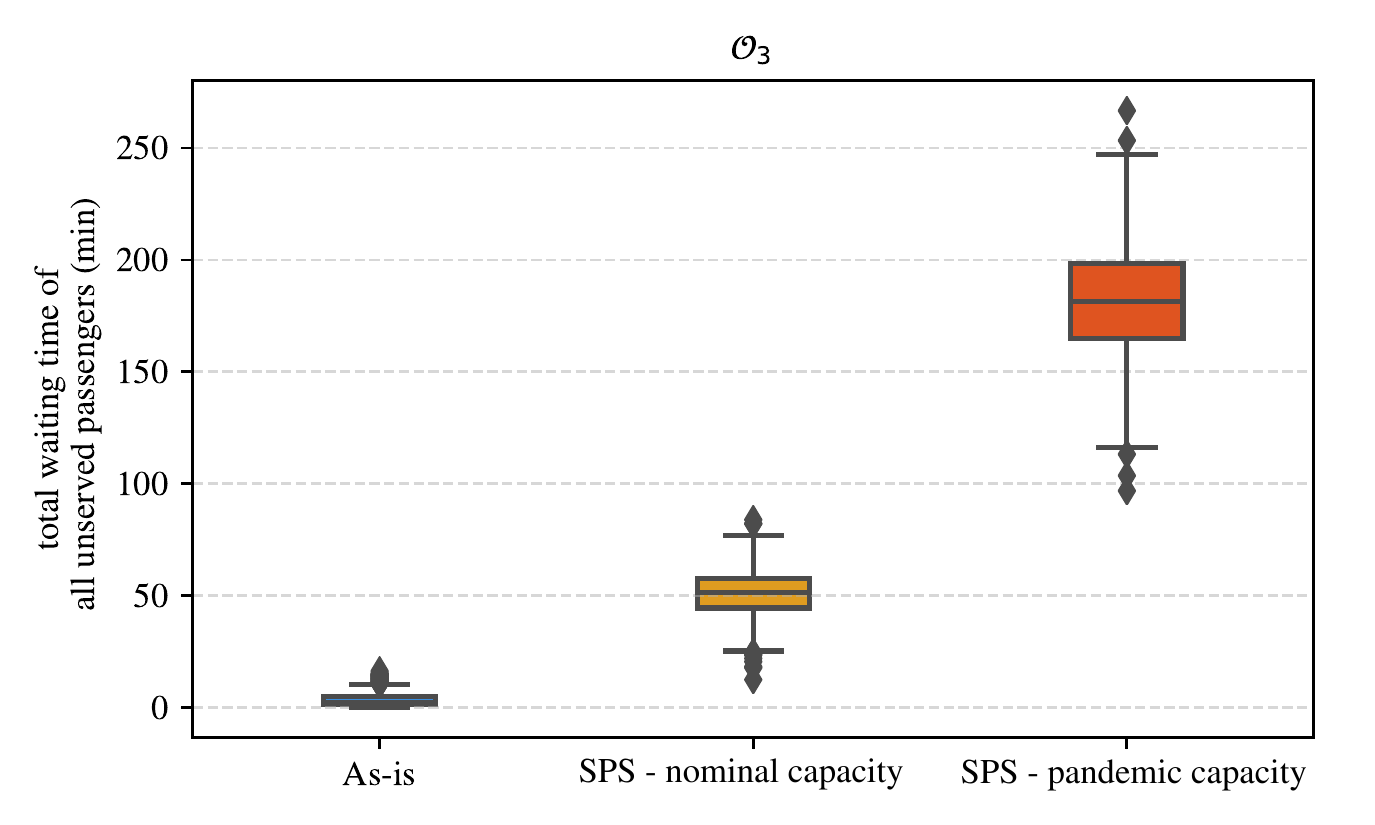}
\caption{Tukey boxplot of the extra passenger waiting times of unserved passengers by vehicle $n$.\label{fig_waiting_times}}
\end{figure}

\section{Concluding Remarks}\label{sec5}
In this work, we developed a service pattern model for determining dynamically the skipped stops of a public transport vehicle which is about to start its service in order to satisfy the pandemic-imposed capacity. Several public transport authorities in major cities have already used service patterns to avoid overcrowding but those decisions are made offline and are not vehicle-specific (e.g., they result in daily stop closures). Our service pattern model fills this research gap and can be applied in dynamic environments by using up-to-date estimations of passenger demand to devise patterns that avoid serving overcrowded stops. One of its benefits is that it can be applied in near real-time for public transport lines with realistic sizes to return an optimal solution.

Although service patterns are used by many public transport operators, one should consider that by skipping particular stops the number of unserved passengers and their waiting times might increase. To evaluate these negative effects resulting from the implementation of a pandemic-driven service pattern, we implement our model in a bus line connecting two cities with the university of Twente. We also implement the as-is service that does not skip any stop and a service pattern model that does not try to meet the pandemic-imposed capacity limit. The results of our evaluation in 1000 scenarios with different passenger demand levels demonstrated that the service pattern that considers the pandemic capacity can reduce the total passenger load that exceeds the pandemic capacity by 188 passengers per vehicle (or 15 passengers per vehicle per stop). Importantly, this service pattern is able to maintain a passenger load beyond the pandemic capacity limit in almost all cases. This, however, results in skipping a considerable number of stops (4 out of 13 stops) and a significant number of unserved passengers that can be up to 50 passengers per trip. 

The aforementioned values are based on the setting of our case study. Our service pattern modern that considers the pandemic-imposed capacity might perform significantly better in public transport lines with relatively low passenger demand because potential skipped stops will not result in many unserved passengers. Additionally, lines with small headways are more suitable for the implementation of service patterns since travelers will not have to wait for an extended period of time if they are skipped by a vehicle. To summarize, our decision support model can propose service patterns for different line services and can help public transport operators to assess the benefits and drawbacks of implementing pandemic-driven service patterns given their operational headways and their passenger demand levels.

In future research, our service pattern model can be expanded to consider more preferences from the operational side of a service line, such as the reduced vehicle travel times due to the skipped stops. One important research topic is also the combination of our model with dynamic frequency models that can increase the number of vehicles at particular time periods of the day to offer an increased vehicle supply to high-demand stops that are skipped by our service patterns. In case of limited vehicle availability from the public transport operator, on-demand services and shared mobility options can complement this supply gap by combining our model with models for on-demand scheduling or vehicle sharing.

\section*{Acknowledgements}
This work if funded by the Dutch Research Council (NWO) under the L4 project `COVID 19 Wetenschap voor de Praktijk', project number: 10430042010018.

\small
\bibliographystyle{trb}
\bibliography{interactcadsample}

\nolinenumbers
\end{document}